\documentclass[a4paper,UKenglish,cleveref, autoref, thm-restate]{lipics-v2021}

\pdfoutput=1 
\hideLIPIcs  

\usepackage{booktabs}       
\usepackage{amsfonts}       
\usepackage{nicefrac}       
\usepackage{microtype}      
\usepackage{lipsum}
\usepackage{fancyhdr}       
\usepackage{graphicx}       
\graphicspath{{media/}}     
\usepackage{todonotes}
\usepackage{algorithm}
\usepackage[noend]{algpseudocode}
\usepackage{stackengine}
\usepackage{multirow,bigdelim}

\mathchardef\mhyphen="2D 
\usepackage{arydshln}
\usepackage{xcolor}
\usepackage{graphicx}
\usepackage{bigstrut}
\usepackage{complexity}
\usepackage{comment}
\usepackage{tikz}
\usepackage{adjustbox}
\usetikzlibrary{positioning} 
\usepackage{amssymb,amsmath,amsthm}
\usepackage{setspace}
\usepackage{mathtools}

\usepackage{algorithmicx,algorithm,eqparbox,array}

\usepackage{thm-restate}

\newcommand{\vect}[1]{\boldsymbol{#1}}

\newcommand{\calY}{\mathcal{Y}}
\newcommand{\ent}[1]{{\Gamma^*_{#1}}}
\newcommand{\aent}[1]{{\overline{\Gamma^*_{#1}}}}
\newcommand{\polym}[1]{{\Gamma_{#1}}}

\newcommand{\mono}[1]{{\mathcal{M}_{#1}}}
\newcommand{\step}[1]{\mathrm{S}_{#1}}

\newcommand{\tup}{\mathrm{Tup}}

\newcommand{\ubound}[1]{\mathrm{Bound}_{#1}}
\newcommand{\logubound}[1]{\mathrm{Log\text{-}Bound}_{#1}}

\newcommand{\Deg}{\mathrm{Deg}}

\newcommand{\Simple}{\mathrm{Simple}}

\newcommand{\modular}[1]{\mathrm{Mod}_{#1}}
\newcommand{\normal}[1]{\mathrm{N}_{#1}}
\newcommand{\monotone}[1]{\mathrm{Mon}_{#1}}

\newcommand{\att}{\mathsf{Var}}
\newcommand{\val}{\mathsf{Val}}

\newcommand{\iip}{{\mathsf{IIP}}}
\newcommand{\iipb}[2]{\iip_{#1}(#2)}
\newcommand{\iips}[1]{\iip_{#1}}

\newcommand{\Dom}{\mathrm{Dom}}

\makeatletter
\def\BState{\State\hskip-\ALG@thistlm}
\makeatother

\title{Information inequality problem over set functions} 

\author{Miika {Hannula\footnote{The author has been supported by the ERC grant 101020762.}}}{University of Helsinki\\Finland\\\texttt{miika.hannula@helsinki.fi}}{}{}{}

\authorrunning{M. Hannula} 

\Copyright{Jane Open Access and Joan R. Public} 

\category{} 

\relatedversion{} 

\acknowledgements{}

\nolinenumbers 

\EventEditors{John Q. Open and Joan R. Access}
\EventNoEds{2}
\EventLongTitle{42nd Conference on Very Important Topics (CVIT 2016)}
\EventShortTitle{CVIT 2016}
\EventAcronym{CVIT}
\EventYear{2016}
\EventDate{December 24--27, 2016}
\EventLocation{Little Whinging, United Kingdom}
\EventLogo{}
\SeriesVolume{42}
\ArticleNo{23}

\begin{document}

\maketitle

\begin{abstract}
Information inequalities appear in many database applications such as query output size bounds, query containment, and implication between data dependencies. 
Recently Khamis et al. \cite{KhamisK0S20} proposed to study the algorithmic aspects of information inequalities, including the information inequality problem: decide whether a linear inequality over entropies of random variables is valid. While the decidability of this problem is a major open question, applications often involve only inequalities
that adhere to specific syntactic forms linked to useful semantic invariance properties. This paper studies the information inequality problem in different syntactic and semantic scenarios that arise from database applications. Focusing on the boundary between tractability and intractability, we show that the information inequality problem is $\coNP$-complete if restricted to normal polymatroids, and in polynomial time if relaxed to monotone functions. We also examine syntactic restrictions related to query output size bounds, and provide an alternative proof, through monotone functions, for the polynomial-time computability of the entropic bound over simple sets of degree constraints.
\end{abstract}

\section{Introduction}
Information inequalities are linear constraints on entropies of random variables.  Often referred to as the laws of information, these inequalities describe what is not possible in information theory. 
More than three decades ago, Pippenger 
 asked whether all such laws follow from the \emph{polymatroidal axioms} \cite{pippenger86}, depicted in Fig. \ref{fig:polym}.
 The polymatroidal axioms are also known to be equivalent to the non-negativity of Shannon's information measures, which consist of entropy, conditional entropy, mutual information, and conditional mutual information. The inequality constraints derivable from the polymatroidal axioms are hence called \emph{Shannon} inequalities.
  Pippenger's question was famously answered in the negative by  Zhang and Yeung who were the first to find a non-Shannon information inequality that is valid over entropies \cite{641561}.
 Zhang and Yeung's proof was based on a novel innovation, identified as the  \emph{copy lemma} in \cite{dougherty2011nonshannon}, which still today remains essentially the only tool to establish novel non-Shannon inequalities \cite{GurpinarR19}.

 Constraints on entropies are  known to have many applications in database theory. 
 Lee \cite{Lee87,Lee87a} observed already in the 80s that database constraints can alternatively be expressed as equalities over information measures. 
  More recently, the implication problem for data dependencies has been connected to validity of information inequalities \cite{KenigS22}, information theory has been used to analyze normal forms in relational and XML data models \cite{ArenasL05}, and query containment for conjunctive queries under bag semantics---a notoriously difficult problem to study---has been proven to be equivalent in certain special cases to checking information inequalities involving maximum. Perhaps the most fruitful application has been the use of information inequalities to obtain tight output size bounds for database queries \cite{AtseriasGM13,GottlobLVV12,GroheM14,KhamisNS16,Khamis0S17,ngoarxiv22}, and the subsequent development of \emph{worst-case optimal join algorithms} that run in time proportional to these bounds \cite{KhamisNS16,Khamis0S17,Ngopods18, NgoPRR18}.

Recently Khamis et al. \cite{KhamisK0S20} initiated the study of the algorithmic properties of information inequalities. The most central problem, called the {information inequality problem}, is to decide whether a given information inequality is valid over all entropic functions. The decidability of this problem is a major open question in the foundations of information theory. It was shown in \cite{KhamisK0S20} that checking the validity of monotone Boolean combinations of information inequalities (including the aforementioned \emph{max-inequalities}) is co-recursively enumerable (co-r.e.). Since the implication problem for conditional independence implication is undecidable \cite{Li23}, validity for general Boolean combinations of information inequalities is known to be undecidable. 
While the focus of \cite{KhamisK0S20} was on generalizations of the information inequality problem, this paper shifts attention to simplifications of the problem. Many applications, such as implication problems or query output size bounds, are related to information inequalities that adhere to specific syntactic forms. 
These syntactic forms  are also often linked to semantic invariance properties which render the associated problems computable and sometimes even tractable. 
Identifying  factors that make the information inequality problem either easy or hard is thus a task that can prove beneficial in multiple application scenarios.

This paper examines the information inequality problem with respect to different syntactic restrictions and semantic settings, focusing in particular on the boundary between tractable and intractable cases. We demonstrate that different factors, including the influence of the coefficients and the expressiveness of the information measures, give rise to $\coNP$-completeness with respect to normal polymatroids (the subset of entropic functions associated with a non-negative I-measure \cite{KenigS22, yeung08}), and disagreement between normal polymatroids and entropic functions. Our findings also reveal that when we relax the semantics to monotone functions or restrict it to modular functions (an implicit result in existing literature), the information inequality problem can be solved in polynomial time. Additionally, we demonstrate that this problem becomes polynomial-time solvable when we impose syntactic restrictions linked to cases where computing the entropic query output size bound is known to be in polynomial time.
Finally, we identify a syntactic restriction over which monotone and entropic functions agree, leading to an alternative proof for the previously established fact \cite{ngoarxiv22} that the entropic bound is polynomial-time computable over simple sets of degree constraints.

\section{Preliminaries}
We write $[n]$ for the set of integers $\{1, \dots, n\}$.
We use boldface letters to denote sets and sequences interchangeably. 
 For two sets (resp. sequences) $\vect{X}$ and $\vect{\vect{Y}}$, we write $\vect{X}\vect{\vect{Y}}$ to denote their union (resp. concatenation). 
If $A$ is an individual element, we sometimes write $A$ instead of $\{A\}$ to denote the singleton set consisting of $A$. 

\subsection{Relational databases}
Fix disjoint countably infinite sets $\att$ and $\val$ of variables and values. Each variable $A \in \att$
is associated with a subset of $\val$, called the \emph{domain} of $A$, denoted $\Dom(A)$. For a vector $\vect{X}=(A_1, \dots ,A_n)$ of variables, 
we write $\Dom(\vect{X})$ for the Cartesian product $\Dom(A_1)\times \dots \times \Dom(A_n)$.
Given a finite set of variables $\vect{X}$, an $\vect{X}$\emph{-tuple} is a mapping $t:\vect{X} \to \val$ such that $t(A) \in \Dom(A)$.
We write $\tup(\vect{X})$ for the set of all $\vect{X}$-tuples. 
 For $\vect{Y}\subseteq \vect{X}$, the \emph{projection $t[\vect{Y}]$ of $t$ on} $\vect{Y}$
is the unique $\vect{Y}$-tuple that agrees with $t$ on $\vect{X}$. 
%
A \emph{relation $R$} over $\vect{X}$ is a subset of $\tup(\vect{X})$.  
The variable set $\vect{X}$ is also called 
\emph{(relation) schema of $R$}. We sometimes write $R(\vect{X})$ instead of $R$ to
emphasize that $\vect{X}$ is the schema of $R$. For $\vect{Y}\subseteq \vect{X}$, 
the \emph{projection} of $R$ on $\vect{Y}$, written $R[\vect{Y}]$, is the set of all projections $t[\vect{Y}]$ where
$t\in R$. 
A \emph{database} is a finite collection of relations $D=\{R^D_1(\vect{X}_1), \dots ,R^D_n(\vect{X}_n)\}$. If $D$ is clear from the context, we drop the superscript $D$ from $R^D_i(\vect{X}_i)$.
We assume in this paper that each relation is finite.

\subsection{Information theory}
Let $X$ be a random variable associated with a finite domain $D=\Dom(X)$ and a probability distribution $p:D \to [0,1]$, where $\sum_{a\in D} p(a)=1$.
The \emph{entropy} of $X$  is defined as
\begin{equation}\label{eq:entropy}
H(X) \coloneqq -\sum_{x\in D} p(x) \log p(x).
\end{equation}
Entropy is non-negative and does not exceed the logarithm of the domain size: $0\leq H(X) \leq \log |D|$. In particular, $H(X)=0$ if and only if $X$ is constant (i.e., $p(a)=1$ for some $a \in D$), and $H(X)=\log |D|$ if and only if $X$ is uniformly distributed (i.e., $p(a) = 1/|D|$ for all $a \in D$). 

 %
 
 In the following, we list some common classes of vectors $\vect{h}=(h_\alpha)_{
 \alpha \subseteq [n]} \in\mathbb{R}^{2^{n}}$ over $n\geq 1$.
 For any vector $\vect{h}$, we write $h(\vect{X}_\alpha)$, where $\vect{X}_\alpha = \{X_i \mid i \in \alpha\}$, to denote the element $h_\alpha$. 
 Thus $\vect{h}$ is also a \emph{set function} over $n$ variables.
 We assume $h(\emptyset)=0$ for all set functions $\vect{h}$.
 For a list of functions $\vect{h}_1, \dots ,\vect{h}_n$, the function $c_1\vect{h}_1+\dots +c_n\vect{h}$ is called a \emph{non-negative combination} (resp. \emph{positive combination}) of  $\vect{h}_1, \dots ,\vect{h}_n$ if $c_i>0$ (resp. $c_i\geq 0$) for all $i\in [n]$.

 \noindent
 \textbf{Entropic functions.}
 Consider now a relation $R$ over a set $\vect{X}=\{X_i\}_{i=1}^n$ of variables with finite domains,
 associated with a probability distribution  $p:R \to [0,1]$. Each subset of $\vect{Y}\subseteq \vect{X}$ can be viewed as 
 a random variable with domain $D=R[\vect{Y}]$ and probability distribution $p_{\vect{Y}}(t) = \sum_{{t'\in R,t'[\vect{Y}]=t}} p(t')$.
 The relation $R$ then defines an \emph{entropic function} 
 $\vect{h}=(h_\alpha)_{
 \alpha \subseteq [n]}$,%
 where $h_\alpha\coloneqq H(\vect{X}_\alpha)$. 
  The \emph{entropic region} $\Gamma^*_n\subseteq \mathbb{R}^{2^{n}}$ consists of all entropic functions over $n$. The \emph{almost entropic region} $\overline{\Gamma^*_n}$ is defined as the topological closure of $\Gamma^*_n$.
%

 \noindent
 \textbf{Polymatroids.} If $\vect{h}$ satisfies the \emph{polymatroidal axioms} (Fig. \ref{fig:polym}), it is called a \emph{polymatroid}.
 The set of polymatroids over $n$  is denoted $\polym{n}$.

  \begin{figure}[ht]\label{fig:axioms}
  \centering
  \begin{tikzpicture}[every node/.style={outer sep=0pt}]
  \def\m{1.4em}
    \node[draw, rounded corners, minimum width=1\textwidth] (box1) {
      \begin{minipage}[t][1.8cm]{0.9\textwidth}
      \hspace{0cm}
        \begin{enumerate}
  \item $h(\emptyset) = 0$  
  \item $h(\vect{X}\cup \vect{Y}) \geq h(\vect{X})$ (monotonicity)
  \item $h(\vect{X}) + h(\vect{Y}) \geq  h(\vect{X}\cap \vect{Y}) + h(\vect{X}\cup \vect{Y})$ (submodularity)
        \end{enumerate}
      \end{minipage}
    };
  \end{tikzpicture}
  \caption{polymatroidal axioms. \label{fig:polym}}
\end{figure}

 \noindent
 \textbf{Monotone functions.} If $\vect{h}$ satisfies the first two axioms 
 of the {polymatroidal axioms}, we call it a \emph{monotone function}, and denote the set of monotone functions over $n$  by $\monotone{n}$.

 \noindent
 \textbf{Normal polymatroids.}
 For $\vect{U}\subseteq \vect{X}$, the function
\begin{equation}\label{eq:step}
s_{\vect{U}}(\vect{W})=\begin{cases}
0 &\text{ if }\vect{W}\subseteq \vect{U};\\
1&\text{ otherwise;}
\end{cases}
\end{equation}
is called a \emph{step function}. We also use the notation $s^{\vect{V}}$ to denote the set function $s_{\vect{X}\setminus \vect{V}}$. The step function $s^{\vect{V}}$ is the uniform distribution of two tuples $t$ and $t'$ such that $t(A)\neq t'(A)$ if and only if $A \in \vect{V}$. The set of all step functions over $n$  is denoted $\step{n}$.  A  \emph{normal polymatroid} is a positive combination of step functions, and the set of all normal polymatroids over $n$ is denoted $\normal{n}$. 

 \noindent
 \textbf{Modular polymatroids.} A polymatroid $\vect{h}$ is called \emph{modular} if the submodularity inequality (Fig. \ref{fig:polym}) is an equality: $h(\vect{X}) + h(\vect{Y}) =  h(\vect{X}\cap \vect{Y}) + h(\vect{X}\cup \vect{Y})$. Alternatively, modular polymatroids can be defined in terms of basic modular functions.
 A \emph{basic modular function} is a step function of the form $s^{\{A\}}$, that is, it is a step function defined in terms of a singleton set. A function $\vect{f}$ is a modular polymatroid if and only if it is a positive combination of basic modular functions.
 

By continuity of Eq. \eqref{eq:entropy}, and since there are no restrictions on domain sizes, $c\vect{h}$ is entropic for $c > 0$ and step functions $\vect{h}$. 
Furthermore, if $\vect{h}$ and $\vect{h}'$ are entropic functions defined by probability distributions $p$ and $p'$ over some relation $R$, the distribution $p''(t \otimes t') \coloneqq p(t)p'(t')$
on the direct product $\{t \otimes t' \mid t,t'\in R\}$,  $(t \otimes t')(X)\coloneqq (t(X),t'(X))$, defines $\vect{h}+\vect{h}'$. This shows that the entropic region is closed under multiplication by positive integers, even though in general it is not closed under positive scalar multiplication \cite{yeung08}; in other words, $\ent{n}$ is not a \emph{cone}. 
We conclude that both modular and normal polymatroids are entropic. 
In fact, Kenig and Suciu \cite{KenigS22} have shown that the normal polymatroids are exactly those entropic functions that have a non-negative I-measure \cite{yeung08}.
The introduced set functions form an increasing sequence
\[
\modular{n}\subseteq \normal{n}\subseteq \ent{n}\subseteq \aent{n} \subseteq \polym{n} \subseteq \monotone{n}.
\]
 If $n \geq 4$, then all the subset relations in this sequence are strict.

We will repeatedly refer to the following Shannon's information measures (over some $\vect{h}$). 
\begin{itemize}
\item Conditional entropy: $h(\vect{Y}\mid \vect{X}) \coloneqq h(\vect{XY})-h(\vect{X})$.
\item Mutual information: $I_{\vect{h}}(\vect{X};\vect{Y}) \coloneqq h(\vect{X})+h(\vect{Y})-h(\vect{X}\vect{Y})$.
\item Conditional mutual information: $I_{\vect{h}}(\vect{Y};\vect{Z}\mid \vect{X}) \coloneqq h(\vect{XY})+h(\vect{XZ})-h(\vect{X})-h(\vect{X}\vect{YZ})$.
\end{itemize}
We may drop the subscript $\vect{h}$ if it is clear from the context.

An \emph{information inequality} is an expression $\phi$ of the form 
\begin{equation}\label{eq:ineq}
c_1h(\vect{X}_1) + \dots +c_kh(\vect{X}_k)\geq 0,
\end{equation}
where $c_i\in \mathbb{R}$, and $\vect{X}_i$ are sets of variables from 
$\{X_j\}_{j=1}^n$. We sometimes write $\phi(\vect{X})$ instead of $\phi$ to emphasize that the set of variables appearing in $\phi$ is $\vect{X}$.
For $V\subseteq \mathbb{R}^{2^n}$, we say that $\phi$ is \emph{valid} over $V$, denoted
$V \models \phi$, if it holds true for all functions $\vect{h} \in V$. 

This paper focuses on the \emph{information inequality problem} ($\iip$), introduced in \cite{KhamisK0S20}, which is to decide whether a given information inequality is valid over $\ent{n}$. 
This problem is co-r.e. \cite{KhamisK0S20}, as the continuity of the entropy \eqref{eq:entropy} and the density of the rationals in the reals imply that enumeration of all rational distributions will eventually lead to a counterexample of \eqref{eq:ineq}, if one exists at all. 
 We introduce the following relativized version of $\iip$. Fixing sets of functions $S_n \subseteq \mathbb{R}^{2^n}$, $n\geq 1$, and a set $\mathcal{C}$ of information inequalities, the \emph{information inequality problem over $S_n$ w.r.t. $\mathcal{C}$} ($\iipb{S_n}{\mathcal{C}}$) is to determine whether a given information inequality $\phi \in \mathcal C$ over $n$ variables is valid over $S_n$. We leave out $S_n$ (resp. $\mathcal{C}$) 
  if
 $S_n=\Gamma^*_n$ (resp. $\mathcal{C}$ contains all information inequalities). Note that an inequality $\phi$ is valid over the entropic region $\ent{n}$ if and only if it is valid over the almost entropic region $\aent{n}$. To see why, $V \models \phi$ is tantamount to $V \subseteq C_\phi$, where $C_\phi = \{\vect{h} \in \mathbb{R}^{2^n}\mid \vect{h}\models \phi\}$, and by taking closures on both sides, $\ent{n}\subseteq C_\phi$ entails $\aent{n} \subseteq C_\phi$.
 More generally, validity over $\ent{n}$ and $\aent{n}$ disagrees with respect to Boolean combinations of information inequalities \cite{KacedR13,KhamisK0S20}.
 Since our focus is on the information inequality problem alone, we now drop the almost entropic region $\aent{n}$ from discussions.

Before proceeding, we  shortly discuss input representation. We assume that the coefficients are rational. Note that in \cite{KhamisK0S20} the inputs of $\iip$ and other related problems are vectors $\vect{c} \in \mathbb{Z}^{2^n}$ representing the coefficients in Eq. \eqref{eq:ineq}. In this paper, we consider the input as a sequence $((c_1, \vect{X}_1), \dots ,(c_k,\vect{X}_k))$, which is potentially exponentially shorter than the aforementioned coefficient vector $\vect{c}$. This distinction is not important if one is solely interested in decidability, as is the case in \cite{KhamisK0S20}. Since our aim is to chart the tractability boundary for different information inequality problems, we opt for the latter more concise representation. Furthermore, we assume that the coefficients themselves are encoded in binary. 

We begin our analysis from intractable examples, and then move on to discuss tractable cases and their connections to query output bounds.

\section{Intractable cases}
Kenig and Suciu \cite{KenigS22} establish an interesting connection between information inequalities and the implication problem for database dependencies. 
Fix a relation schema $\vect{X}$ of $n$ variables.
  An expression of the form $\sigma=(\vect{V};\vect{W}\mid \vect{U})$ is called a \emph{conditional independence} (CI). 
 If $\vect{UVW}=\vect{X}$, $\sigma$ is specifically called a \emph{saturated conditional independence} (SCI), and 
 if $\vect{V}=\vect{W}$, it is called a \emph{conditional} and shortened as $(\vect{V}\mid \vect{U})$.   
 %
The results in \cite{KenigS22} entail that if $\Sigma$ is a set of SCIs and conditionals, and $\tau$ is a conditional, then for any $V$ such that  $\step{n}\subseteq V\subseteq \polym{n}$,
 \begin{equation}\label{eq:imp}
 V \models \sum_{\sigma\in \Sigma} h(\sigma)\geq h(\tau) \iff \Sigma \models \tau,
 \end{equation}
  where the right-hand side denotes implication between corresponding MVDs and FDs over database relations. Whether or not $\sum_{\sigma\in \Sigma} h(\sigma)\geq h(\tau)$ is valid over $V$ can be thus decided in polynomial time, because the implication problem for MVDs and FDs is known to be in polynomial time \cite{BeeriFH77}.
 
There are at least two ways to make the inequality in Eq. \eqref{eq:imp} harder. One possibility is to allow more complex information measures,  much like how one can allow more expressive database dependencies in the implication problem. For instance, once the aforementioned syntactic restrictions are lifted, the implication problem for CIs becomes undecidable in both database theory (where CIs are known as embedded multivalued dependencies) and probability theory \cite{herrmann95,Li23}. Another possibility,  which does not seem to have a counterpart in the implication problem, is to permit coefficients distinct from $1$.
Next, we consider both of these strategies in isolation, considering  first complex information measures.

 The mutual information of two random variables generalizes to the {multivariate mutual information} over a set of random variables $\vect{S}$. For a general set function $\vect{h}$, the \emph{multivariate mutual information} is given as
\begin{equation}\label{eq:multi-mutual}
I_h(\vect{S}) = \sum_{\vect{T} \subseteq \vect{S}} (-1)^{|\vect{T}|-1} h(\vect{T}).
\end{equation}
Again, we drop the subscript whenever this is possible without confusion. A particular case of the multivariate mutual information is the three-variate one: 
\begin{equation}\label{eq:three}
I(ABC) = h(A) + h(B) + h(C) - h(AB) - h(AC) - h(BC) + h(ABC).
\end{equation}
 Multivariate mutual information is non-negative on step functions, but it can be negative on entropic functions.
 For example, if $A$ and $B$ are independent and uniformly either $0$ or $1$, and $C=A+B \pmod 2$, then Eq. \eqref{eq:three} evaluates to $-1$ on the corresponding entropic function. Next we show that solving inequalities containing three-variate mutual informations and conditional entropies can already be $\coNP$-hard, even if each coefficient is exactly one. The result, proven by reduction from monotone satisfiability, holds for step functions but does not extend to entropic functions.
 

A conjunctive normal form Boolean formula $\phi$ is called \emph{monotone} if each clause in $\phi$ contains only negative or only positive literals.
The \emph{monotone satisfiability problem} is the problem of deciding whether such a formula $\phi$ has a satisfying truth assignment. This problem is 
 well known to be $\NP$-complete \cite{Gold78}, and it remains $\NP$-complete even if each clause consists of exactly three distinct literals
\cite{Li97a}. Let us denote this restriction of the problem by \textsf{3DMONSAT}.
An instance of \textsf{3DMONSAT} can be represented as a pair $\phi=(\phi^+, \phi^-)$, where $\phi^+$ (resp. $\phi^-$) is the set of all positive (resp. negative) clauses of $\phi$, and each clause is a set of exactly $3$ variables.

\begin{theorem}\label{thm:strogn}
The information inequality problem over normal polymatroids is $\coNP$-complete.
\end{theorem}
\begin{proof}
Since normal polymatroids are positive combinations of step functions, and inequalities are preserved under positive combinations, the information inequality problems over step functions and normal polymatroids coincide.
The upper bound is thus obvious. 
For the lower bound,
we present a reduction from the complement of \textsf{3DMONSAT} to the information inequality problem over step functions. Let $\phi=(\phi^+, \phi^-)$ be an instance of \textsf{3DMONSAT}. Suppose $\vect{X}$ is the set of variables appearing in $\phi$.
We may assume without loss of generality that every satisfying assignment must map at least one variable to $1$.

Define an information inequality
\begin{equation}\label{eq:reduction}
\sum_{\vect{C}\in \phi^+} h(\vect{X} \mid \vect{C}) + \sum_{\vect{C}\in \phi^-} I(\vect{C}) \geq h(\vect{X}),
\end{equation}
where $I(\vect{C})$ is the three-variate mutual information \eqref{eq:three} over the variables of $\vect{C}$, and $h(\vect{X} \mid \vect{C})$ is the conditional entropy of $\vect{X}$ given $\vect{C}$.

 Each subset $\vect{Y}\subseteq \vect{X}$ determines a unique step function $s^{\vect{Y}}$ (Eq. \ref{eq:step}) and a unique Boolean assignment 
 \[
 m_{\vect{Y}}(A)=\begin{cases}
 1 & \text{ if }A \in \vect{Y},\\
 0 & \text{ otherwise.}
 \end{cases}
 \]
 We claim that $m_{\vect{Y}}$ satisfies $\phi$ if and only if Eq. \eqref{eq:reduction} is false for $h=s^{\vect{Y}}$.
 
 Assume first that $m_{\vect{Y}}$ satisfies $\phi$. By our assumption some variable is mapped to $1$, which means that $\vect{Y}$ is non-empty. In particular, $s^{\vect{Y}}(\vect{X})=1$. For any positive clause $\vect{C}\in \phi^+$, we have $\vect{C}\cap \vect{Y}\neq \emptyset$, and consequently $s^{\vect{Y}}(\vect{X} \mid \vect{C}) = s^{\vect{Y}}(\vect{X})- s^{\vect{Y}}(\vect{C})=0$.
For any negative clause $\vect{C}\in \phi^-$, we have $\vect{C}\not\subseteq \vect{Y}$, in which case it is straightforward to verify that $I(\vect{C})=0$. We conclude that Eq. \eqref{eq:reduction} is false for $h=s^{\vect{Y}}$.

Assume then that $m_{\vect{Y}}$ does not satisfy $\phi$. If $s^{\vect{Y}}(\vect{X})=0$, then Eq. \eqref{eq:reduction} is trivially true for $h=s^{\vect{Y}}$ by the non-negativity of the conditional entropy and the multivariate mutual information on step functions.
Suppose then  $s^{\vect{Y}}(\vect{X})\neq 0$, in which case $s^{\vect{Y}}(\vect{X})= 1$.
Assuming $m_{\vect{Y}}$ does not satisfy some $\vect{C}\in \phi^+$, we have $\vect{C}\cap \vect{Y}=\emptyset$ implying $s^{\vect{Y}}(\vect{X} \mid \vect{C})=0$. Assuming $m_{\vect{Y}}$ does not satisfy some $\vect{C}\in \phi^-$, we have $\vect{C}\subseteq \vect{Y}$ implying $I(\vect{C})=1$. We conclude that, for $h=s^{\vect{Y}}$, the left-hand side of Eq. \eqref{eq:reduction} is at least $1$, and thus the inequality holds. This concludes the proof of the claim.

The claim implies that $\phi$ is not satisfiable if and only if Eq. \eqref{eq:reduction} is valid over step functions. The theorem statement follows, since the reduction is clearly polynomial.
\end{proof}
Observe that Eq. \eqref{eq:reduction} is syntactically similar to the inequality in Eq. \eqref{eq:imp} in that each coefficient is exactly one. The difference comes from allowing three-variate mutual information, whereas the inequality in Eq. \eqref{eq:imp} allows only specific forms of conditional mutual information. The above proof moreover establishes \emph{strong} $\coNP$-completeness, because the problem remains $\coNP$-complete even under unary encoding of the coefficients.

Alternatively, the preceding theorem can be proven by reducing $3$-colorability to inequalities that allow the coefficients to grow while using only conditionals. Let $G=(\vect{V},\vect{E})$ be a graph consisting of a vertex set $\vect{V}$ and a set of undirected edges $\vect{E}$.
For each node $A \in \vect{V}$, introduce variables $A_r,A_g,A_b$ representing possible colors of $A$. Assume that the graph contains $n$ vertices. We define
\begin{equation}\label{eq:color}
\sum_{\substack{c\in \{r,g,b\}\\A \in \vect{V}}} h(A_c) + 
\sum_{\substack{c,d \in \{r,g,b\}\\c\neq d\\A\in \vect{V}}} (2n+1)h(\vect{V}\mid A_cA_d)+\sum_{\substack{c \in \{r,g,b\}\\\{A,B\}\in \vect{E}}} (2n+1)h(\vect{V}\mid A_cB_c)
\geq (2n+1) h(\vect{V}).
\end{equation}
Then $G$ is three-colorable if and only if Eq. \eqref{eq:color} is not valid over step functions (Appendix \ref{sect:alt}). This way of proving Theorem \ref{thm:strogn} establishes also strong $\coNP$-completeness, since each coefficient is bounded by a polynomial in the input size.

It is necessary to allow coefficients other than $1$ in Eq. \eqref{eq:color}.  Otherwise, the equivalence \eqref{eq:imp} holds, meaning that 
the validity problem is equivalent to the implication problem for FDs, which is in polynomial time.
We may also notice that Eq. \eqref{eq:color} behaves differently for step functions and entropic functions, even though both  functions are non-negative on all the occurring information measures; in contrast, the proof of Theorem \ref{thm:strogn} relied on three-variate mutual information which is only guaranteed to be non-negative for step functions but can be negative for entropic functions.
 Suppose $G$ is the complete graph of four vertices.
Since $G$ is not three-colorable, Eq. \eqref{eq:color} is valid over  step functions, and remains valid even if one extends $G$ by an arbitrary number of isolated  vertices. However, once this extension is sufficiently large, Eq. \eqref{eq:color} can be shown to be  not valid for entropic functions by selecting the uniform distribution over the Boolean support.


Having shown two ways to construct strongly $\coNP$-complete information inequality problems, the next theorem presents a specific syntactic form of Eq. \eqref{eq:ineq} that admits a reduction from a weakly $\coNP$-complete problem. The proof is located in Appendix \ref{sect:weak}. 
\begin{restatable}{theorem}{weak}\label{thm:weak}
The information inequality problem over normal polymatroids w.r.t. inequalities of the form 
\begin{equation}\label{eq:ineq2}
c_1h(\vect{X}_1) + \dots +c_kh(\vect{X}_k)\geq d_1h(\vect{Y}_1) + \dots +d_lh(\vect{Y}_l) \quad (c_i,d_i>0),
\end{equation}
where $|\vect{Y}_i|\leq 2$, is $\coNP$-complete.
\end{restatable}

\section{Tractable cases}
We have seen that intractable inequalities arise from (i) complex information measures even if coefficients are restricted to $1$, (ii) more simple information measures if coefficients are allowed to grow, and (iii) inequalities where the negative coefficients are associated with sets of size at most two. In this section we 
consider restrictions that give rise to inequalities solvable in polynomial time. We are specifically interested in inequalities of the form 
$
\sum_{\sigma \in \Sigma} w_i h(\sigma) \geq h(\vect{X})
$
where $w_i\geq 0$, and $\Sigma$ is a set of conditionals. Such inequalities make appearance when information theory is applied to obtain tight upper bounds for query output sizes. 
Since inequalities of the form \eqref{eq:color} are intractable, imposing syntactic restrictions on $\Sigma$ becomes necessary.

We next introduce information-theoretic query upper bounds, after which we move on to discuss the complexity of related syntactic restrictions and semantic modifications of the information inequality problem. 

\subsection{Query upper bounds}
 Fix a relation $R$ over a variable set $\vect{X}$. Given vectors $\vect{U} , \vect{V}$ of variables from $\vect{X} $, and values $\vect{u}\in \Dom(\vect{U})$, the $V$\emph{-degree} of $\vect{U} = \vect{u}$ in $R$, denoted $\deg_R(\vect{V} \mid \vect{U} = \vect{u})$, is the number of distinct values $\vect{V}$ that occur in $R$ together with the value $\vect{u}$ of $\vect{U}$. The $\max$-$\vect{V}$\emph{-degree} of $\vect{U}$, denoted $\deg_R(\vect{V} \mid \vect{U})$ is the maximum $V$-degree of $\vect{U} = \vect{u}$ over all $\vect{u}$. Expressions of the form $\deg_R(\sigma)\leq B$ (omitting the parentheses of $\sigma$), where $\sigma $ is a conditional and $B \geq 1$, are usually called \emph{degree constraints}.
 Note that $\deg_R(\vect{V} \mid \vect{U})=1$ if and only if $R$ satisfies the functional dependency $\vect{U} \to \vect{V}$.
A $\Sigma$\emph{-inequality} is an information inequality $\phi_{\Sigma}(\vect{X},\vect{w})$ of the form
\begin{equation}\label{eq:sigeq}
\sum_{\sigma \in \Sigma} w_{\sigma} h(\sigma) \geq h(\vect{X}),
\end{equation}
where $\vect{w}=(w_{\sigma})_{\sigma \in \Sigma}$ is a sequence of non-negative reals.

Fix a \emph{self-join-free full conjunctive query}, i.e., a quantifier-free first-order formula of the form
\[
Q(\vect{X}) = R_1(\vect{X}_1)\land \dots \land R_n(\vect{X}_n),
\]
where $R_i(\vect{X}_i)$ are \emph{relational atoms} over distinct \emph{relation names} $R_i$, and variable sequences $\vect{X}_i$  such that $\vect{X} = \bigcup_{i=1}^n\vect{X}_i$. Note that this incurs a slight abuse of notation, because we also use $R_i$ to denote relations. 
We say that a set of conditionals $\Sigma$ is \emph{guarded} by $Q$ if every $\sigma = (\vect{V} \mid \vect{U})$ from $\Sigma$ is associated with a relation name $R_i$, called the \emph{guard} of $\sigma$ and denoted $R_{\sigma}$, such that $\vect{U}\vect{V} \subseteq\vect{X}_i $. 
A sequence of the form $\vect{B}=(B_{\sigma})_{\sigma \in \Sigma}$, $B_\sigma \geq 1$, form the \emph{degree values} associated with $\Sigma$. 
A database $D$ containing relations $R_\sigma$, $\sigma \in \Sigma$, \emph{satisfies} a conditionals-values pair  $(\Sigma, \vect{B})$, written $D\models (\Sigma, \vect{B})$, if $\deg_{R_\sigma}(\sigma)\leq B_\sigma$ for all $\sigma \in \Sigma$.

For a set $S \subseteq \mathbb{R}^{2^n}$, and a set of conditionals $\Sigma$ guarded by $Q(\vect{X})$ and associated with values $\vect{B}$, define the \emph{bound} of $Q$ w.r.t. $\Sigma,S,\vect{B}$ as
 \[
 \ubound{S}(Q,\Sigma,\vect{B}) \coloneqq \inf_{\substack{\vect{w}\geq 0\\\S\models \phi_{\Sigma}(\vect{X},\vect{w})}} \prod_{\sigma \in \Sigma} B^{w_\sigma}_{\sigma}.
\]
The bounds $\ubound{\modular{n}}, \ubound{\normal{n}}, \ubound{\ent{n}}, \ubound{\polym{n}}$ are often referred to as the \emph{modular bound}, the \emph{coverage bound}, the \emph{entropic bound}, and the \emph{polymatroid bound}.
 Writing $Q(D)$ for the output of $Q$ on a database $D=\{R_1(\vect{X}_1), \dots ,R^D_n(\vect{X}_n)\}$, one can prove that the entropic bound is valid: $|Q(D)|\leq \ubound{\ent{n}}(Q,\Sigma,\vect{B})$ whenever $D \models (\Sigma, \vect{B})$ (see, e.g., \cite{Suciu23} for a derivation of the bound). Since $\ent{n}\subseteq \polym{n}$, the entropic bound is less than or equal to the polymatroid bound. The entropic bound is asymptotically tight, but it is open whether or not the bound is computable. The polymatroid bound can be attained by solving a linear program of exponential size, but it is not tight. 
 
Fortunately, there are well-behaving syntactic restrictions for sets of conditionals $\Sigma$, some of which are presented next. 
\begin{itemize}
\item  For $\sigma = (\vect{V} \mid \vect{U})$, where $\vect{U}=\emptyset$, degree constraints of the form $\deg_R(\sigma)\leq B$
are called \emph{cardinality constraints}. 
 The AGM bound \cite{AtseriasGM13} can be viewed as the entropic bound over  a specific set  of cardinality constraints.
\item  $\Sigma$ is called \emph{acyclic} if the following directed graph is acyclic: the vertices are the variables in $\vect{X}$, and there is an edge from $A$ to $B$ if $A \in \vect{X}$ and $B\in \vect{Y}\setminus \vect{X}$, for some $(\vect{Y}\mid \vect{X})\in \Sigma$. 
\item  $\Sigma$ is called \emph{simple} if $|\vect{U}|\leq 1$ for each  $(\vect{V} \mid \vect{U})\in \Sigma$. 
\end{itemize}
The entropic bound is polynomial-time computable in all of these cases. Let us call the $\Sigma$-inequality \eqref{eq:sigeq} \emph{acyclic} (resp. \emph{simple}) if the underlying set $\Sigma$ is acyclic (resp. simple).  The sets of conditionals underlying cardinality constraints are vacuously acyclic, and validity for acyclic $\Sigma$-inequalities coincides for modular functions, entropic functions, and polymatroids \cite{KhamisKNS21}. Consequently, the entropic bound becomes computable in polynomial time through a linear program describing the validity of Eq. \eqref{eq:sigeq} over basic modular functions.
Validity for simple $\Sigma$-inequalities similarly coincides for  entropic functions,  polymatroids, and normal polymatroids. This does not immediately entail that the entropic bound for simple $\Sigma$ is computable in polynomial time, because normal polymatroids are constructed with step functions, and there are exponentially many step functions in the number of variables.
The entropic bound is nevertheless known to be polynomial-time computable in this case, as has been shown recently  \cite{ngoarxiv22}.

 Eq. \eqref{eq:color} can now be viewed as an $\Sigma${-inequality} \eqref{eq:sigeq} (up to scaling) arising from $\Sigma$ that does not belong to any of the aforementioned well-behaving classes. Since validity of inequalities of the form Eq. \eqref{eq:color} is $\coNP$-hard  over step functions (Appendix \ref{sect:alt}),  
 this immediately gives us the following result.
 \begin{theorem}\label{thm:deg}
 The information inequality problem over normal polymatroids w.r.t. $\Sigma$-inequalities 
 is $\coNP$-complete. This problem remains $\coNP$-hard even if $|\vect{U}|\leq 2$ for all $(\vect{V}\mid\vect{U})\in \Sigma$.
\end{theorem}
Related to the previous result, computing the coverage bound over a set of conditionals is known to be $\NP$-hard, and computing the polymatroid bound over an arbitrary set of conditionals can be efficiently reduced to computing the polymatroid bound over another set of conditionals $(\vect{V}\mid \vect{U})$ such that $|\vect{U}|\leq 2$ and $|\vect{V}|\leq 3$ \cite{ngoarxiv22}.

We now turn to discuss tractable cases of the information inequality problem obtained either by syntactic restrictions or semantic modifications. The syntactic restrictions we consider correspond quite closely to the aforementioned acyclic/simple $\Sigma$-inequalities.
\subsection{Modular functions}\label{sect:modular}
Since modular functions can be constructed as positive combinations of basic modular functions, 
  the information inequality problem is trivially polynomial-time computable in this context.
\begin{proposition}
The information inequality problem over modular functions is in polynomial time.
\end{proposition}
One example of a syntactic class with respect to which validity over entropic functions corresponds to validity over modular functions are the acyclic $\Sigma$-inequalities.
Given an acyclic set $\Sigma$ of conditionals $(\vect{V}\mid \vect{U})$, and a polymatroid $\vect{h}$ over $\vect{X}$, one can construct a modular function $\vect{f}$ such that (i) $f(\vect{X})=h(\vect{X})$, and (ii) $f(\vect{V}\mid \vect{U}) \leq h(\vect{V}\mid \vect{U})$ for all $(\vect{V}\mid \vect{U})\in \Sigma$  \cite{Ngopods18}.  
Consequently, validity of acyclic $\Sigma$-inequalities \eqref{eq:sigeq} coincides for polymatroids, entropic functions, and modular functions. Thus it is known that all the aforementioned bounds (modular, coverage, entropic, and polymatroid bounds) coincide and are polynomial-time computable if  $\Sigma$ is acyclic \cite{Ngopods18}. With respect to the information inequality problem, we analogously obtain the following result.
\begin{proposition}
Let $\modular{n} \subseteq K \subseteq \polym{n}$.
The information inequality problem over $K$ w.r.t. acyclic $\Sigma$-inequalities is in polynomial time.
\end{proposition}

\subsection{Monotone functions}
Next we will show that, at the other extreme direction, the information inequality problem over monotone functions is also in polynomial time. 
\begin{theorem}\label{thm:monptime}
The information inequality problem over monotone functions is in polynomial time.
\end{theorem}
Analogously to the previous section, this semantic modification of the information inequality problem helps us  identify syntactic classes with respect to which the general information inequality problem is tractable. 
Before we proceed into details, let us give a short sketch of the proof of this theorem.
We associate an information inequality \eqref{eq:ineq} with 
a \emph{set representation}
$(S^+, S^-)$, where 
\begin{align*}
S^+ \coloneqq& \{(\vect{X}_i,k) \mid k \in [|c_i|], c_i >0 \}\text{, and}\\
S^- \coloneqq& \{(\vect{X}_i,k) \mid k \in [ |c_i|], c_i <0 \}.
\end{align*}
 Then,  we present a fixed-point algorithm (Alg. \ref{alg:mono}) to capture validity of information inequalities over monotone functions. This algorithm iteratively decomposes  an input inequality into monotonicity axioms. For this, it maintains a bipartite directed graph $G$ initialized as $G_S = (S^+ \cup S^-, E)$, where $E$ is the set of edges from $S^+$ to $S^-$ that correspond to  possible monotonicity axioms (\emph{forward edges}).  The initial graph contains no edges from $S^-$ to $S^+$ (\emph{backward edges}). The number of these backward edges, which represent those monotonicity axioms that are currently selected for the decomposition,  is increased in each iteration. 
  %
 Although the algorithm as such does not run in polynomial time (it runs in \emph{pseudo-polynomial time}, i.e., in polynomial time in the length of the input and the numeric values of the coefficients), it does guide us toward a characterization of valid inequalities as positive combinations of monotonicity axioms and \emph{non-negativity axioms} $h(\vect{X})\geq 0$, which are derivable as combinations of the first two polymatroidal axioms. These combinations are polynomial in the input length, and consequently can be found through a linear program of polynomial size, which entails the desired result.

\algrenewcommand{\algorithmiccomment}[1]{\hfill// \eqparbox{COMMENT\thealgorithm}{#1}}
\algnewcommand{\LongComment}[1]{\hfill// \begin{minipage}[t]{\eqboxwidth{COMMENT\thealgorithm}}#1\strut\end{minipage}}

\begin{algorithm}
  \caption{Decomposition algorithm for inequalities. \label{alg:mono}}
   \hspace*{\algorithmicindent} \textbf{Input: }Set representation $S=(S^+,S^-)$ of $\phi$  \\[2pt]
   \hspace*{\algorithmicindent} \textbf{Output:} true iff $\phi$ is $\mono{n}$-valid
  \begin{algorithmic}[1]
  \State $G \gets G_{S}, S_0 \gets S^+, S_1 \gets S^-$ \vskip 3pt
  \While {$G$ contains a path $u_0, \dots ,u_m$ from $S_0$ to $S_1$} \vskip 3pt
  \State remove $u_0$ from $S_0$ and $u_m$ from $S_1$
  \vskip 3pt
  \State remove from $G$ (backward) edges $(u_1,u_2), (u_3, u_4), \dots ,(u_{m-2}, u_{m-1})$ \hspace{-.5cm}
  \State add to $G$ (backward) edges $(u_1,u_0), (u_3, u_2), \dots ,(u_{m}, u_{m-1})$
  \EndWhile \vskip 3pt
  \Return true if $S_1$ is empty, otherwise false
  \end{algorithmic}
  \end{algorithm}
  The following example demonstrates the use of Alg. \ref{alg:mono}.

  \begin{example}
Consider an information inequality of the form 
\begin{equation}\label{eq:example}
\vect{X}\vect{Y} + \vect{Y}\vect{Z} + 2\vect{X}\vect{Z} + \vect{X} \geq \vect{Y} + 3\vect{Z}.
\end{equation}
 The set representation is
$(S^+, S^-)$, where 
\begin{align*}
S^+ =& \{(\vect{X}\vect{Y},1), (\vect{Y}\vect{Z},1), (\vect{X}\vect{Z},1), (\vect{X}\vect{Z},2), (\vect{X},1)\}, \text{ and}\\
S^-=&\{(\vect{Y},1), (\vect{Z},1), (\vect{Z},2), (\vect{Z},3)\}.
\end{align*}
Clearly, the inequality \eqref{eq:example} is valid over monotone functions. Alg. \ref{alg:mono} also returns true after four iterations. The leftmost graph in Fig. \ref{fig:tikz} illustrates the starting point for the last iteration in one possible implementation. The edges from right to left (backward edges) represent monotonicity axioms that have been selected in the previous iteration. The edges from left to right (forward edges), some of which are visible in the middle graph of Fig. \ref{fig:tikz}, represent possible monotonicity axioms.
Since there is a path from $S_0$ to $S_1$, we can increase the number of selected monotonicity axioms by deleting the backward edges in the path, and  changing the direction of the forward edges in the path.
 The rightmost graph illustrates the result of this modification. Since $S_1$ becomes empty, the algorithm terminates returning true. The final state of the algorithm represents an integral decomposition of Eq. \eqref{eq:example} into monotonicity and non-negativity axioms.
\end{example}

\begin{figure}[h!]
\begin{center}
\scalebox{.85}{
\begin{tabular}{ccccc}
\hspace{-10mm}
\belowbaseline[0pt]{
\begin{tikzpicture}
\def\leftn{5}
\def\rightn{5}
    \node [rectangle, draw] (left1) at (0,1) {$(\vect{X},1)$};
     \node [rectangle, draw] (left2) at (0,2) {$(\vect{X}\vect{Y},1)$};
    \node [rectangle, draw] (left3) at (0,3+.2) {$(\vect{Y}\vect{Z},1)$};
        \node [rectangle, draw] (left4) at (0,4+.2) {$(\vect{X}\vect{Z},1)$};
    \node [rectangle, draw] (left5) at (0,5+.2) {$(\vect{X}\vect{Z},2)$};
    \node [rectangle, draw] (right2) at (2.5,2) {$(\vect{Z},1)$};
    \node [rectangle, draw] (right3) at (2.5,3+.2) {$(\vect{Y},1)$};
        \node [rectangle, draw] (right4) at (2.5,4+.2) {$(\vect{Z},2)$};
    \node [rectangle, draw] (right5) at (2.5,5+.2) {$(\vect{Z},3)$};
\foreach \x in {3,...,5}
    \draw [->] (right\x) -- (left\x);
\draw[-, dashed, gray] (-1, 2.6) -- (1, 2.6);
\draw [dashed, gray] (-1,0.2) rectangle (1,\leftn+.9) node[anchor=south west] at (-1,0.2) {$S_0$};
\draw [dashed, gray] (-1,0.2) rectangle (1,\leftn+.9) node[anchor=north, yshift=0pt] at (0,\leftn+1.5) {$S^+$};
\draw[-, dashed, gray] (1.5, 2.6) -- (3.5, 2.6);
\draw [dashed, gray] (1.5,1.2) rectangle (3.5,\rightn+.9) node[anchor=south east] at (3.5,1.2) {$S_1$};
\draw [dashed, gray] (1.5,1.2) rectangle (3.5,\rightn+.9) node[anchor=north, yshift=0pt] at (2.5,\rightn+1.5) {$S^-$};
\end{tikzpicture}
}
        &
        \belowbaseline[85pt]{
$\xrightarrow{\text{}}$
}
        &
\belowbaseline[0pt]{
\begin{tikzpicture}
\def\leftn{5}
\def\rightn{5}
    \node [rectangle, draw] (left1) at (0,1) {$(\vect{X},1)$};
     \node [rectangle, draw] (left2) at (0,2) {$(\vect{X}\vect{Y},1)$};
    \node [rectangle, draw] (left3) at (0,3+.2) {$(\vect{Y}\vect{Z},1)$};
        \node [rectangle, draw] (left4) at (0,4+.2) {$(\vect{X}\vect{Z},1)$};
    \node [rectangle, draw] (left5) at (0,5+.2) {$(\vect{X}\vect{Z},2)$};
    \node [rectangle, draw] (right2) at (2.5,2) {$(\vect{Z},1)$};
    \node [rectangle, draw] (right3) at (2.5,3+.2) {$(\vect{Y},1)$};
        \node [rectangle, draw] (right4) at (2.5,4+.2) {$(\vect{Z},2)$};
    \node [rectangle, draw] (right5) at (2.5,5+.2) {$(\vect{Z},3)$};
\foreach \x in {3,...,5}
    \draw [->] (right\x) -- (left\x);
\draw [->] (left2.north east) -- (right3.south west);
\draw [->] (left3.north east) -- (right4.south west);
\draw [->] (left4.south east) -- (right2.north west);
\draw[-, dashed, gray] (-1, 2.6) -- (1, 2.6);
\draw [dashed, gray] (-1,0.2) rectangle (1,\leftn+.9) node[anchor=south west] at (-1,0.2) {$S_0$};
\draw [dashed, gray] (-1,0.2) rectangle (1,\leftn+.9) node[anchor=north, yshift=0pt] at (0,\leftn+1.5) {$S^+$};
\draw[-, dashed, gray] (1.5, 2.6) -- (3.5, 2.6);
\draw [dashed, gray] (1.5,1.2) rectangle (3.5,\rightn+.9) node[anchor=south east] at (3.5,1.2) {$S_1$};
\draw [dashed, gray] (1.5,1.2) rectangle (3.5,\rightn+.9) node[anchor=north, yshift=0pt] at (2.5,\rightn+1.5) {$S^-$};
\end{tikzpicture}
}
&
        \belowbaseline[85pt]{
$\xrightarrow{\text{}}$
}
          &
\belowbaseline[0pt]{
\begin{tikzpicture}
\def\leftn{5}
\def\rightn{5}
    \node [rectangle, draw] (left1) at (0,1) {$(\vect{X},1)$};
     \node [rectangle, draw] (left2) at (0,2+.2) {$(\vect{X}\vect{Y},1)$};
    \node [rectangle, draw] (left3) at (0,3+.2) {$(\vect{Y}\vect{Z},1)$};
        \node [rectangle, draw] (left4) at (0,4+.2) {$(\vect{X}\vect{Z},1)$};
    \node [rectangle, draw] (left5) at (0,5+.2) {$(\vect{X}\vect{Z},2)$};
    \node [rectangle, draw] (right2) at (2.5,2+.2) {$(\vect{Y},1)$};
    \node [rectangle, draw] (right3) at (2.5,3+.2) {$(\vect{Z},2)$};
        \node [rectangle, draw] (right4) at (2.5,4+.2) {$(\vect{Z},1)$};
    \node [rectangle, draw] (right5) at (2.5,5+.2) {$(\vect{Z},3)$};
\foreach \x in {2,...,5}
    \draw [->] (right\x) -- (left\x);
\draw[-, dashed, gray] (-1, 1.6) -- (1, 1.6);
\draw [dashed, gray] (-1,0.2) rectangle (1,\leftn+.9) node[anchor=south west] at (-1,0.2) {$S_0$};
\draw [dashed, gray] (-1,0.2) rectangle (1,\leftn+.9) node[anchor=north, yshift=0pt] at (0,\leftn+1.5) {$S^+$};
\draw [dashed, gray] (1.5,1.6) rectangle (3.5,\rightn+.9) node[anchor=south east] at (3.5,1) {$S_1=\emptyset$};
\draw [dashed, gray] (1.5,1.6) rectangle (3.5,\rightn+.9) node[anchor=north, yshift=0pt] at (2.5,\rightn+1.5) {$S^-$};
\end{tikzpicture}
}
\end{tabular}
}
\end{center}
\caption{Last iteration of Alg. \ref{alg:mono} for Eq. \eqref{eq:example}. \label{fig:tikz}}
\end{figure}
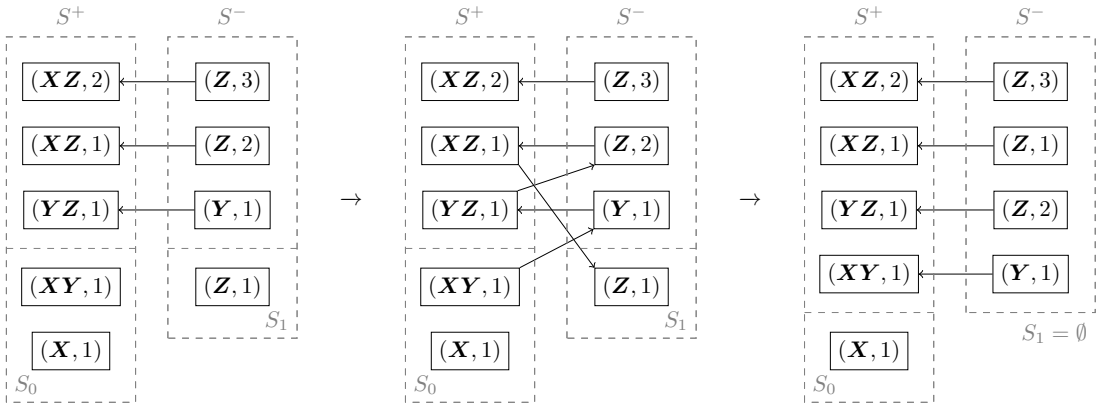

An inequality $\phi$ of the form \eqref{eq:ineq} can be identified with its \emph{coefficient function} $\vect{c}_{\phi}$, where $c_{\phi}(\vect{X})= c$ if the term $ch(\vect{X})$ appears in \eqref{eq:ineq}, and otherwise $c_{\phi}(\vect{X})=0$. We then say that $\phi$ is a (positive) combination of inequalities $\phi_1,\dots,\phi_n$ if $\vect{c}_{\phi}$ is a (positive) combination of $\vect{c}_{\phi_1}, \ldots ,\vect{c}_{\phi_n}$.
We furthermore say that a combination of functions $c_1\vect{h}_1 + \dots + c_n \vect{h}_n$ is \emph{separable} if there exist no $i,j$ and $\vect{Y}$ such that ${h_i}(\vect{Y})<0$ and ${h_j}(\vect{Y})>0$, while $c_i\neq 0 \neq c_j$. This definition is extended to combinations of inequalities in the natural way.
 For example, any positive combination of $h(A) + h(B) \geq 0$ and  $h(B) -h(AB) \geq 0$ is separable, but no positive combination of inequalities $h(A) - h(B) \geq 0$ and $h(A) +h(B) \geq 0$ is separable. In particular, if $\phi$ is a positive and separable combination of monotonicity and non-negativity axioms, then $h(\vect{X})$ cannot appear in the left-hand side of $\psi_0$ and in the right-hand side of $\psi_1$, for any two axioms $\psi_0$ and $\psi_1$  appearing in the combination.

We also say that a set function $\vect{h}$ is Boolean-valued if it maps every $\vect{X}$ to either $0$ or $1$. The proof the following lemma is located in Appendix \ref{sect:alg}.
\begin{restatable}{lemma}{alg}\label{lem:alg}
Let $\phi$ be an information inequality of the form
\begin{equation}\label{eq:inf2}
c_1h(\vect{X}_1) + \dots +c_kh(\vect{X}_k)\geq 0 \quad (c_i\in \mathbb R).
\end{equation}
The following are equivalent:
\begin{enumerate}
\item $\phi$ is valid over monotone functions. 
\item $\phi$ is valid over monotone, Boolean-valued functions. 
\item $\phi$ is a positive and separable combination of  monotonicity and non-negativity axioms.
\end{enumerate}
\end{restatable}
\vspace{3mm}

Since linear programming is in polynomial time, we can now establish Theorem \ref{thm:monptime} as a consequence of the following lemma. 
This lemma will also be applied in the next section that focuses on simple $\Sigma$-inequalities. 

\begin{lemma}\label{lem:matrix}
For each information inequality $\phi$ of the form
\begin{equation}\label{eq:matrix}
c_1h(\vect{X}_1) + \dots +c_kh(\vect{X}_k)\geq d_1h(\vect{Y}_1) + \dots +d_lh(\vect{Y}_l) \quad (c_i,d_i>0),
\end{equation}
 there exists a matrix 
  $M$ 
  such that the inequality $M\vect{x} \geq \vect{cd}$ has a solution $\vect{x}\geq 0$ if and only if $\phi$ is valid over monotone functions. In particular, $M$ can be constructed in polynomial time from $\phi$ (with rational coefficients).
\end{lemma}
\begin{proof}
 Consider first a set $\vect{Y}_i$ from the right-hand side of the inequality. Let $\vect{X}_{i_1}, \dots ,\vect{X}_{i_m}$ list all those sets $\vect{X}_j$ from the left-hand side that contain $\vect{Y}_i$ as a subset. We need to describe how the term $d_ih(\vect{Y}_i)$ is distributed to monotonicity axioms. For this, define
\begin{equation}\label{eq:mondist}
x^i_{i_1} + \dots + x^i_{i_m}  \geq d_i,
\end{equation}
where $x^i_j$ is a variable denoting the coefficient of the monotonicity axiom $h(\vect{X}_j)\geq h(\vect{Y}_i)$.
We also need to ensure that this variable does not grow exceedingly large. Consider a set $\vect{X}_j$ from the left-hand side of the inequality, and let $\vect{Y}_{j_1}, \dots ,\vect{Y}_{j_n}$ list all those sets from the right-hand side that are contained in $\vect{X}_j$.
\begin{equation}\label{eq:large}
x^{j_1}_{j} + \dots + x^{j_n}_{j} \leq c_j.
\end{equation}
Combining Eqs. \eqref{eq:mondist} and \eqref{eq:large} we obtain an inequality $M\vect{x} \geq \vect{cd}$, where $M$ is a $((k+l)\times (kn))$-matrix with  entries of from $-1,0,1$. Obviously $M$ can be constructed in polynomial time given $\phi$. Moreover, $M\vect{x} \geq \vect{cd}$ has a solution $\vect{x}  \geq \vect{0}$ if and only if $\phi$ is a positive and separable combination of  monotonicity and non-negativity axioms. 
 The statements of the theorem then follow by Lemma \ref{lem:alg}.
\end{proof}


\subsection{Simple inequalities}
Let us first recall the reason why normal and general polymatroids 
 are known to agree on the validity of simple $\Sigma${-inequalities}. 
 On the one hand, any such $\Sigma${-inequality} over a variable set $\vect{X}$ can be presented in the form 
\begin{equation}\label{eq:validity}
c_1 h(\vect{X}_1) + \dots +  c_nh(\vect{X}_n) \geq d_1h(\vect{Y}_1)+ \dots + h(\vect{Y}_m) \quad (c_i,d_i>0),
\end{equation}
where each $\vect{Y}_i$ is either the full set $\vect{X}$ or some singleton set $\{X\}$. On the other hand, every polymatroid $\vect{h}$ over $\vect{X}$ can be associated with a normal polymatroid $\vect{f}$ such that $\vect{f}\leq \vect{h}$, ${f}(\vect{X})=h(\vect{X})$, and $f(X)=h(X)$ for all $X \in \vect{X}$ \cite{Suciu23}.
If $\vect{h}$ is a counterexample for Eq. \eqref{eq:validity}, then $\vect{f}$ must also be a counterexample. Since normal polymatroids are positive combinations
step functions, it follows that at least one step function is also  a counterexample.  This brings us to the following result. 
\begin{theorem}[\cite{Suciu23}]\label{thm:suciu}
Let $\phi(\vect{X})$ be an information inequality of the form Eq. \eqref{eq:validity},
where each $\vect{Y}_i$ is either the full set $\vect{X}$ or a singleton set. Then, $\phi$ is valid over step functions if and only if it is valid over polymatroids.
\end{theorem}
Since $\step{n}\subseteq \normal{n}\subseteq \ent{n}\subseteq \polym{n}$, it follows that validity for simple $\Sigma${-inequalities} coincides for step functions, normal polymatroids, entropic functions, and polymatroids. If we remove terms of the form $h(\vect{X})$ from the right-hand side of Eq. \eqref{eq:validity}, the previous result extends to monotone functions. 
\begin{theorem}\label{lem:unary}
Let $\phi(\vect{X})$ be an information inequality of the form Eq. \eqref{eq:validity},
where each $\vect{Y}_i$ is a singleton set. Then, $\phi$ is valid over step functions if and only if it is valid over monotone functions.
\end{theorem}
\begin{proof}
Since step functions are monotone, we only need to consider the ``only-if'' direction. To show the contraposition, assume that $\phi$ is not valid over monotone functions. By Lemma \ref{lem:alg}, we find a monotone, Boolean-valued function $h$ such that Eq. \eqref{eq:validity} becomes false. Consider the step function $s^{ \vect{U}}$, where $\vect{U}$ consists of all those variables $A_i$ that are mapped to $1$ by $h$. Clearly, $h$ and $s^{\vect{U}}$ agree on the right-hand side of Eq. \eqref{eq:validity}. Furthermore, for any set $\vect{Z}$, we have $s^{ \vect{U}}(\vect{Z}) \leq h(\vect{Z})$ by monotonicity of $h$. Consequently, Eq. \eqref{eq:validity} is also false for $s^{ \vect{U}}$, meaning that $\phi$ is not valid over step functions.
\end{proof}
It follows that validity for information inequalities of the form \eqref{eq:validity},
where $\vect{Y}_i$ are singletons, is decidable in polynomial time with respect to any $K$ such that $\step{n} \subseteq K \subseteq \monotone{n}$, including $K=\ent{n}$.
Note that simple $\Sigma$-inequalities are not of this form; rewritten in the form \eqref{eq:validity} one of the sets $\vect{Y}_i$ is the full variable set. However, as we will see next, it is possible to remove such terms in a single step.


 Continuing our analysis of $\phi(\vect{X})$ of the form  \eqref{eq:validity}, fix  a variable $A$ from $\vect{X}$. Define sums
 \[
 c_A = \sum_{\substack{i\in [n]\\A \in \vect{X}_i}} c_i \quad\text{ and }\quad d_A = \sum_{\substack{i\in [n]\\A \in \vect{Y}_i}}d_i,
 \]
 and define the $A$-\emph{reduction} of $\phi$ as the inequality $\phi^{A}(\vect{X}\setminus \{A\})$ given as
 \begin{equation}\label{eq:validity2}
(c_A-d_A)h(\vect{X}\setminus\{A\}) + \sum_{\substack{i\in [n]\\A \notin \vect{X}_i}} c_ih(\vect{X}_i) \geq \sum_{\substack{i\in [n]\\A \notin \vect{Y}_i}} d_ih(\vect{Y}_i),
\end{equation}

\begin{lemma}\label{lem:redconj}
An information inequality $\phi(\vect{X})$ of the form \eqref{eq:validity} (having no restrictions on sets $\vect{Y}_i$)
 is valid over step functions if and only if
for all $A \in \vect{X}$, the $A$-reduction $\phi^{A}$ of $\phi$ is valid over step functions.
\begin{enumerate}
\item $c_A \geq d_A$, and
\item  the $A$-reduction $\phi^{A}(\vect{X}\setminus \{A\})$ is valid over step functions. 
\end{enumerate}
\end{lemma}
\begin{proof}
 Note that $s^A 
 \models \phi $ if and only if $c_A \geq d_A$. Also, 
 if $\emptyset \neq \vect{Y}\subseteq \vect{X}\setminus \{A\}$,  we have $s^{\vect{Y}\cup\{A\}} \models \phi$ if and only if $s^{\vect{Y}}\models \phi^A$, where $s^{\vect{Y}\cup\{A\}}$ and $s^{\vect{Y}}$ refer specifically to the set functions over $\vect{X}$ and $\vect{X}\setminus \{A\}$, respectively. The statement of the lemma follows.
\end{proof}
In particular, if each $\vect{Y}_i$ in the inequality \eqref{eq:validity} is either a singleton or the full set $\vect{X}$, then checking validity of this inequality reduces to checking validity of a linear number of inequalities in which the sets appearing in the right-hand side are all singletons. Theorems \ref{thm:monptime} and \ref{lem:unary}, and Lemma \ref{lem:redconj} thus immediately give us the following corollary.
\begin{corollary}
The information inequality problem w.r.t. simple $\Sigma$-inequalities is in polynomial time.
\end{corollary}
One may recall from Theorem \ref{thm:deg} that, at least in the context of step functions, the requirement of $\Sigma$ being simple is necessary. 

We conclude this section by offering an alternative proof for the fact that the entropic bound for simple sets of conditionals $\Sigma$ is polynomial-time computable. In order to formulate this statement precisely, we need the concept of the logarithmic bound. Similarly to the degree values, the \emph{log-degree values} associated with $\Sigma$ are defined as a sequence $\vect{b}=(b_{\sigma})_{\sigma \in \Sigma}$, where $b_\sigma \geq 0$. A function $\vect{h}$ \emph{satisfies}  $(\Sigma, \vect{b})$, denoted $\vect{h}\models(\Sigma, \vect{b})$,   if $h(\sigma)\leq b_{\sigma}$ for all $\sigma \in \Sigma$. 
For a set $S \subseteq \mathbb{R}^{2^n}$, and a set of conditionals $\Sigma$ guarded by a query $Q(\vect{X})$ and associated with values $\vect{b}$, define the \emph{log-bound} of $Q$ w.r.t. $S$ as
 \[
 \logubound{S}(Q,\Sigma,\vect{b}) \coloneqq \inf_{\substack{\vect{w}\geq 0\\S\models \phi_{\Sigma}(\vect{X},\vect{w})}} \sum_{\sigma \in \Sigma} w_\sigma  b^{\sigma},
\]
where $\phi_{\Sigma}(\vect{X},\vect{w})$ is the $\Sigma$-inequality \eqref{eq:sigeq}.
It is known that the \emph{entropic log-bound} $\logubound{\ent{n}}$ is computable in polynomial time \cite{ngoarxiv22}. In the following, we present an alternative proof for this fact  via monotone functions.

\begin{theorem}
Let  $\Sigma$ be a set of conditionals that is guarded by a query $Q(\vect{X})$ and associated with values $\vect{b}$. If $\Sigma$ is simple,
 the entropic log-bound $\logubound{\ent{n}}(Q,\Sigma,\vect{b})$ is computable in polynomial time in the size of the input $(Q,\Sigma,\vect{b})$.
\end{theorem}
\begin{proof}
We construct a linear program that is polynomial in the size of the input and such that its optimal value is attained at the entropic log-bound.
Theorem \ref{thm:suciu} entails 
\begin{equation}\label{eq:first}
  \ent{n}\models\phi_{\Sigma}(\vect{X},\vect{w}) \iff \step{n}\models \phi_{\Sigma}(\vect{X},\vect{w}),
\end{equation}
where $\phi_{\Sigma}$ is the $\Sigma$-inequality \eqref{eq:sigeq}.
Lemma \ref{lem:redconj} implies that 
\begin{equation}\label{eq:second}
\step{n}\models \phi_{\Sigma}(\vect{X},\vect{w})\iff \forall A \in \vect{X}:
 c_A \geq d_A\text{ and }\step{n-1} \models \phi^A_{\Sigma},
\end{equation}
where $c_A , d_A$ are the sums of coefficients $w_\sigma$ computed from $\phi_\Sigma$ for a variable $A$.
Since $\phi^A_{\Sigma}$ contain only singletons on their right-hand sides, Lemma \ref{lem:unary} yields
\begin{equation*}
{\step{n-1}}\models \phi^A_\Sigma \iff {\monotone{n-1}}\models \phi^A_\Sigma.
\end{equation*}
By Theorem \ref{thm:monptime} we can construct in polynomial time matrices $M_A$ such that 
\begin{equation*}
 {\monotone{n-1}}\models \phi^A_\Sigma \iff M_A \vect{x}_A \geq \vect{w}_A \text{ for some }\vect{x}_A\geq 0,
\end{equation*}
where $\vect{w}_A$ is a list (with possible repetitions) of coefficients $w_{\sigma}$ that appear in $\phi^A_{\Sigma}$. 
Note that we should now treat $w_{\sigma}$ as variables, since we are interested in optimizing their values.
Thus we rewrite $M_A \vect{x}_A \geq \vect{w}_A \land c_A \geq d_A$ as $M'_A \vect{x}_A\vect{w}_A \geq 0$, where $M'_A$ is obtained from $(M_A\mid -I_{|\vect{w}_A|})$ by adding one extra row to describe the inequality $c_A \geq d_A$.
Then we construct a single matrix $M^*$ such that $M^*\vect{x}\vect{w} = (\vect{x}_A\vect{w}_A)_{A\in \vect{X}}$, where $\vect{w}= (w_\sigma)_{\sigma \in \Sigma}$, and $\vect{x}$ is the concatenation of all $\vect{x}_A$. Finally, composing $M'_A$ diagonally into a single matrix $M_{\vect{X}}$, and writing $M_\sigma=M_{\vect{X}}M^*$, we obtain
\begin{equation}\label{eq:third}
\forall A \in \vect{X}:
 c_A \geq d_A\text{ and }\step{n-1} \models \phi^A_{\Sigma} \iff M_\Sigma \vect{x}\vect{w} \geq 0.
\end{equation}
By Eqs. \eqref{eq:first}, \eqref{eq:second},  and \eqref{eq:third}
we obtain
\[
\logubound{\ent{n}}(Q,\Sigma,\vect{b}) =\inf_{\substack{\vect{w}\geq 0\\\ent{n}\models \phi_{\Sigma}(\vect{X},\vect{w})}} \sum_{\sigma \in \Sigma} w_\sigma  b^{\sigma}= \min_{\substack{\vect{xw}\geq 0\\M_\Sigma \vect{x}\vect{w} \geq 0}} \sum_{\sigma \in \Sigma} w_\sigma  b^{\sigma}.
\]
Since $M_\Sigma$ can be constructed in polynomial time in the size of $(Q,\Sigma,\vect{b})$, we can compute in polynomial time the entropic log-bound $\logubound{\ent{n}}(Q,\Sigma,\vect{b})$ as the
optimal value of the linear program
\begin{alignat*}{9}
\text{minimize} &\quad  \sum_{\sigma \in \Sigma} w_\sigma  b^{\sigma} & \quad &  \\
\text{subject to} &\quad
M_\Sigma \vect{x}\vect{w}     \geq \vect{0} & \quad &\\
 & \quad \quad \hspace{2.1mm}  \vect{xw} \geq \vect{0}
\end{alignat*}
\end{proof}

\usetikzlibrary{shapes.geometric, calc}

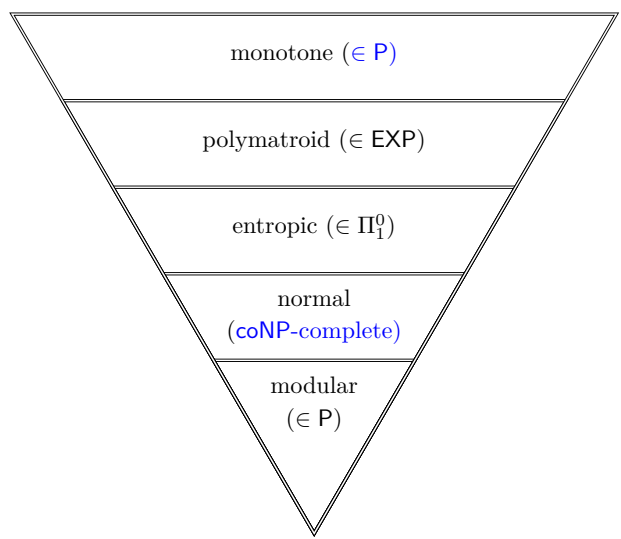
\begin{figure}
\centering
\scalebox{.9}{
\begin{turn}{180}
\begin{tikzpicture}
  \coordinate (sharedtip) at (0,0); 
  \foreach \i in {2,3,4,5,6} {
    \pgfmathsetmacro\size{0.5*\i}
    \ifthenelse{\i = 2}{
    \node[draw,regular polygon, regular polygon sides=3, double, minimum size=\i *1.7cm,rotate=0, label={[yshift=.62cm]below:\rotatebox{180}{modular}}] at (sharedtip) {};
    }
    {
    	    \ifthenelse{\i = 3}{
		 \node[draw,regular polygon, regular polygon sides=3, double, minimum size=\i *1.7cm,rotate=0, label={[yshift=.6cm]below:\rotatebox{180}{normal}}] at (sharedtip) {};
		 }
		 {    	    \ifthenelse{\i = 4}{
			 \node[draw,regular polygon, regular polygon sides=3, double, minimum size=\i *1.7cm,rotate=0, label={[yshift=.9cm]below:\rotatebox{180}{entropic ($\in \Pi^0_1$)}}] at (sharedtip) {};
			 }
			 {\ifthenelse{\i = 5}{
			 			 \node[draw,regular polygon, regular polygon sides=3, double, minimum size=\i *1.7cm,rotate=0, label={[yshift=.9cm]below:\rotatebox{180}{polymatroid ($\in \EXP$)}}] at (sharedtip) {};
						 }
						 {
						 			 \node[draw,regular polygon, regular polygon sides=3, double, minimum size=\i *1.7cm,rotate=0, label={[yshift=.9cm]below:\rotatebox{180}{monotone ({\color{blue}$\in \P$)}}}] at (sharedtip) {};
			 
						 }
			 }
		 }
    }
    \coordinate (sharedtip) at ($(sharedtip) + (0, -.85cm)$); 
  }
  \coordinate (sharedtip) at ($(sharedtip) + (0, 6*.85cm)$); 
  \node at (0,-1.26) {\rotatebox{180}{({\color{blue}$\coNP$-complete)}}};
    \node at (0,0.03) {\rotatebox{180}{($\in \P$)}};
\end{tikzpicture}
\end{turn}
}
\caption{Complexity of the information inequality problem over normal/modular/ordinary polymatroids as well as monotone/entropic functions. The coloring denotes the results obtained in this paper. \label{fig:kooste}}
\end{figure}

\section{Conclusion}
The present paper marks the first attempt to demarcate the tractability boundary for different variants of the information inequality problem, introduced in \cite{KhamisK0S20}.
 We established that this problem is $\coNP$-complete over normal polymatroids, and in polynomial time over monotone functions (see Fig. \ref{fig:kooste} for a summary).
 Restricted to $\Sigma$-inequalities where $|\vect{U}|\leq 2$ for all $(\vect{V}\mid \vect{U})\in \Sigma$, we proved that the information inequality problem remains $\coNP$-hard over normal polymatroids. The same problem was shown to be in polynomial time over normal polymatroids, entropic functions, and polymatroids   if $|\vect{U}|\leq 1$. If every set in the right-hand side of Eq. \eqref{eq:validity} is a singleton or the full variable set,  we proved that the information inequality problem is in polynomial time over any $K$ that falls inbetween normal polymatroids and monotone functions. Using this result, we constructed an alternative proof for the polynomial-time computability of the entropic bound in the case where the set of conditionals $\Sigma$ is simple.

Several questions remain unanswered, however. To the best of the author's knowledge, no lower bounds have been established for $\iip$ over polymatroids or entropic functions. For example, it may be possible to prove an upper bound for $\iip$ over polymatroids that is better than the exponential time bound obtained through linear programming.
We known at least that every valid inequality over polymatroids is a positive combination of the polymatroidal axioms \cite{yeung08}. The problem is that the number of sets needed in the combination may be exponential in the number of variables. Theorem \ref{thm:monptime} demonstrates that if the space defined by an exponential number of axioms is simple enough, then every valid inequality decomposes into a positive combination in which no more than a polynomial number of sets appear.

\bibliography{lipics-v2021-sample-article,biblio}

\appendix

\section{Alternative $\coNP$-hardness proof}\label{sect:alt}
Recall that validity coincides for step functions and normal polymatroids, and thus it suffices to consider validity in the former sense.
We reduce from three-colorability. Let $G=(\vect{V},\vect{E})$ be a graph consisting of a vertex set $\vect{V}$ and a set of undirected edges $\vect{E}$.
For each node $A \in \vect{V}$, we introduce variables $A_r,A_g,A_b$ representing possible colors of $A$. Assume that the graph contains $n$ vertices. We define
\begin{equation}\label{eq:color2}
\sum_{\substack{c\in \{r,g,b\}\\A \in \vect{V}}} h(A_c) + 
\sum_{\substack{c,d \in \{r,g,b\}\\c\neq d\\A\in \vect{V}}} (2n+1)h(\vect{V}\mid A_cA_d)+\sum_{\substack{c \in \{r,g,b\}\\\{A,B\}\in \vect{E}}} (2n+1)h(\vect{V}\mid A_cB_c)
\geq (2n+1) h(\vect{V}).
\end{equation}
We claim that $G$ is three-colorable if and only if Eq. \eqref{eq:color2} is not valid over $\step{n}$. 

Assume first  Eq. \eqref{eq:color2} is not valid, and let $s_{\vect{U}}$, $\vect{U}\subseteq \vect{V}$, be a
step function such that Eq. \eqref{eq:color2} is false for $h=s_{\vect{U}}$. 
 We claim that the function that maps each vertex $A$ to a color $c$ if $A_c \in \vect{U}$ is well-defined and constitutes a coloring of the graph.
Since the entropy and the conditional entropy
are non-negative for all step functions, we have $s_{\vect{U}}(\vect{V}) = 1$, and thus the right-hand side of Eq. \eqref{eq:color2} is $2n+1$. Consequently, the left-hand side is at most $2n$. From the first summation term, we obtain that $\vect{U}$ must contain at least $n$ elements. Moreover, each term of the form $h(\vect{V}\mid A_cA_d)$ or $h(\vect{V}\mid A_cB_c)$ must be zero. In particular, we have $A_cA_d\not\subseteq \vect{U}$ and $A_cB_c\not\subseteq \vect{U}$, which entails that each vertex is assigned exactly one color, and  no two vertices connected by an edge are assigned the same color. We conclude that the function defined by the step function is well defined and constitutes a graph coloring.

Assume then  Eq. \eqref{eq:color2} is valid. For each coloring of the vertices we may define a subset $\vect{U}\subseteq \vect{V}$ such that $A_c \in \vect{U}$ if and only if vertex $A$ is assigned color $c$. Then, the first summation term in the left-hand side of Eq. \eqref{eq:color2} is $n$, and the second summation term is zero. By hypothesis, some term of the form $h(\vect{V}\mid A_cB_c)$
must be non-zero, which means that there exists an edge whose endpoints are assigned the same color. This concludes the proof of the claim. 

Since the reduction is in polynomial time, and each coefficient is bounded by a polynomial in the input size,  strong $\coNP$-completeness again follows.\qed

\section{Reduction from weakly $\coNP$-complete problem}\label{sect:weak}
\weak*
\begin{proof}
Since the upper bound is immediate, we focus on the lower bound.
The \emph{partition problem} is to decide whether a given multiset of positive integers $I$ can be partitioned to two submultisets $I_0$ and $I_1$ such that the sum of the integers in $I_0$ equals that the sum of the integers in $I_1$. This problem is $\NP$-complete, but only in the weak sense, as it can be solved in  polynomial time in the length of the input and the sum of the absolute values of the integers. We reduce from the complement of the partition problem.

Suppose $I=\{\{x_1, \dots,x_n\}\}$ is the input multiset (we use double brackets to differentiate a multiset from a set).
Without loss of generality the sum $\sum_{i=1}^n x_i$ of the integers in $I$ is even.
Each $x_i$ is a positive integer that possibly appears in $I$  multiple times. 
Let $m= x_1 + \dots +x_n$ be the total sum of the integers in $I$. Let $A_i$ be a variable for each $i\in [n]$, and denote by $\vect{X}$ the set of all variables. Consider an information inequality of the form
\begin{equation}\label{eq:partition}
((m/2)^2-1) h(\vect{X}) \geq \sum_{\substack{i,j\in [n]\\i\neq j}} x_ix_j(h(A_i\mid A_j)+h(A_j\mid A_i)).
\end{equation}
Note that each step function $s_{\vect{U}}$, $\vect{U}\subseteq \vect{V}$, corresponds to a partitioning of $I$. Moreover, $h(A_i\mid A_j)+h(A_j\mid A_i)=1$ if $A_i$ and $A_j$ appear on opposite sides of the partitioning, and otherwise $h(A_i\mid A_j)+h(A_j\mid A_i)=0$. The right-hand side of Eq. \eqref{eq:partition} thus becomes 
\[
(\sum_{A_i \in \vect{U}} x_i)(\sum_{A_i \in \vect{X}\setminus \vect{U}} x_i),
\]
which equals $(m/2)^2$ if $\sum_{\substack{i\in [n] \\A_i \in \vect{U}}}x_i = \sum_{\substack{i\in [n] \\A_i \in \vect{X}\setminus \vect{U}}}x_i $, and is otherwise bounded from above by $(m/2)^2-1$. It follows that Eq. \eqref{eq:partition} is not valid if and only if $I$ is a ``yes'' instance of the partition problem.
\end{proof}

\section{Completeness of fixed-point algorithm}\label{sect:alg}
\alg*
\begin{proof}
The implications $(3)\Rightarrow (1)$ and $(1)\Rightarrow (2)$ are immediate. We prove that $(2)\Rightarrow (3)$.  Clearly, if  this implication holds w.r.t. $c_i\in \mathbb Z$, then it holds w.r.t. $c_i\in \mathbb Q$.  We first prove the following claim. 
\begin{claim}
If the implication $(2)\Rightarrow (3)$ holds w.r.t. $c_i\in \mathbb Q$, then it holds w.r.t. $c_i\in \mathbb R$. 
\end{claim}
\begin{claimproof}
To prove this, assume $\phi$ is valid over $\monotone{n}^{0,1}$.
Let $(\phi_n)$ be a sequence of information inequalities
\begin{equation}\label{eq:sequence}
c^n_1h(\vect{X}_1) + \dots +c^n_kh(\vect{X}_k)\geq 0 \quad (c^n_i\in \mathbb Q, n\geq 1),
\end{equation}
where $\lim_{n \to \infty} c_i^n=c_i$ and $c_i^{n} \geq c_i^{n+1} $. We may assume that $c_i^1$ is negative if $c_i$ is negative. That is, $(c^n_i)$ is a sequence of positive (resp. negative) values if $c_i$ is positive (resp. negative).
Clearly, if $\phi$ is valid over Boolean-valued, monotone functions, then so are $\phi_n$. By hypothesis, $\phi_n$ decompose into positive and separable combinations of  monotonicity and non-negativity axioms.
Writing $\vect{c}_{\phi}$ for the coefficient function arising from $\phi$, 
 we may write
\begin{equation}\label{eq:sepa}
\vect{c}_{\phi_n}= d^n_1\vect{c}_{\psi_1} + \dots + d^n_m\vect{c}_{\psi_m} \quad (d^n_i \geq 0),
\end{equation}
where $\psi_l$ list all possible monotonicity and non-negativity axioms respectively of the form $h(\vect{X}_i)\geq 0$ and $h(\vect{X}_i)- h(\vect{X}_j) \geq 0$, where $i,j\in [k]$ and $\vect{X}_j\subseteq \vect{X}_i$, excluding those $\psi_l$ for which the coefficient $d^n_l$ is always zero.
That is, the combinations \eqref{eq:sepa} are separable and have fixed length over all $n\geq 1$;
 recall that separability was defined with respect to terms having a non-zero coefficient.
%
Fix attention to an arbitrary $\psi_l$ being either of the form $h(\vect{X}_i)\geq 0$ or $h(\vect{X}_i)- h(\vect{X}_j) \geq 0$.
In this case, the coefficient function $\vect{c}_{\psi_l}$ maps $\vect{X}_i$ to $1$, that is,  $c_{\psi_l}(\vect{X}_i)=1$.
We claim that the coefficient $c_i$ of $h(\vect{X}_i)$ in Eq. \eqref{eq:inf2} is positive. 
For this, consider some $p\geq 1$ such that the coefficient $d^p_l$ of $\vect{c}_{\psi_l}$ is strictly positive. Assume toward contradiction that $c_i$ is not positive, meaning that it is negative.
Then by construction, $c^p_i$ is negative (i.e., $c_{\phi_p}(\vect{X}_i)<0$), whence $c_{\psi_{l'}}(\vect{X}_i)<0$ for some $l'\neq l$ associated with a strictly positive coefficient $d^p_{l'}$. Since $c_{\psi_l}(\vect{X}_i)>0$, this contradicts  separability of  \eqref{eq:sepa}, proving our claim.
The claim entails by construction that $c^n_i$ are positive 
 for all $n\geq 1$. Hence we obtain $d^n_l \leq c^n_i \leq c^1_i$ by separability of  \eqref{eq:sepa}. 

We conclude that $(\vect{d}_n)= (d^n_1, \dots ,d^n_m)$ is an infinite and bounded sequence of $\mathbb R^m$. The Bolzano-Weierstrass theorem entails that $(\vect{d}_n)$ has a subsequence $(\vect{d}_{n_p})$ that converges to some $\vect{d}=(d_1, \dots ,d_m)$. Obviously the vector $\vect{d}$ is non-negative.
By continuity,
\[
\vect{c}_{\phi} = \lim_{p \to \infty}\vect{c}_{\phi_{n_p}} = \lim_{p \to \infty} d^{n_p}_1\vect{c}_{\psi_1} + \dots + d^{n_p}_m\vect{c}_{\psi_m} = d_1\vect{c}_{\psi_1} + \dots + d_m\vect{c}_{\psi_m}.
\]
The obtained combination is separable, because otherwise some combination \eqref{eq:sepa} is not separable for large enough $n_p$, which leads to a contradiction. We conclude that $\phi$ is a positive and separable combination of monotonicity and non-negativity axioms, which shows that  $(2)\Rightarrow (3)$ w.r.t. $c_i\in \mathbb R$.
\end{claimproof}

It remains to prove that $(2)\Rightarrow (3)$ w.r.t. $c_i\in \mathbb Z$.
  Let $S=(S^+,S^-)$ be the set representation of
  $\phi$. 
  %
  We associate $S$ with a directed
  graph $G_{S}$, where
  \begin{itemize}
    \item the set of nodes are the elements of $S^+$ and $S^-$, and
    \item there is a directed edge from $(\vect{X},i)$ to $(\vect{Y},j)$ if $(\vect{X},i)\in S^+$, $(\vect{Y},j)\in S^-$, and $\vect{Y} \subseteq \vect{X}$.
  \end{itemize}
  
  Consider Alg. \ref{alg:mono} which maintains a bipartite directed graph $G$ that is initially set up as $G_{S}$.
  The monotonicity axioms
  isolated at the current step are represented as
  directed edges going from $S^-$ to $S^+$ \emph{backward edges}; in the beginning no such edges have been introduced yet.
  The edges that proceed from $S^+$ to $S^-$ (\emph{forward edges}) are kept fixed.

We say that a node $u$ is \emph{connected} to a node $v$ in a directed graph if $u=v$, or there is a sequence of nodes (a \emph{path} from $u$ to $v$) in which the first node is $u$, the last node is $v$, 
   and each node is connected to the following node by a directed edge. A set  $\vect{U}$ is \emph{connected} to another set $\vect{V}$ is some node in $\vect{U}$ is connected to some node in $\vect{V}$.

  Consider the following claim. 
  \begin{claim}
  If Alg. \ref{alg:mono} returns true on $\phi$, then $\phi$ is a positive and separable combination of the monotonicity and non-negativity axioms.
  \end{claim}
  \begin{claimproof}
  Consider the graph $G$ and 
  the sets $S_0$ and $S_1$ after termination of the algorithm.
    Note that if $G$ contains a backward edge $(u,v)$, then the reverse edge $(v,u)$ forms a forward edge of $G_{S}$ and consequently corresponds to a monotonicity axiom.
  The backward edges also form a bijection from $S^-\setminus S_1$ to $S^+\setminus S_0$.
  Since $S_1$ is empty by assumption, and
   $S_0$ can be viewed as representing
   non-negativity axioms, it can now be observed that $\phi$ decomposes into a positive and separable combination of monotonicity and non-negativity axioms. 
   This proves the claim.
\end{claimproof}

We now prove the contraposition of $(2)\Rightarrow (3)$ w.r.t. $c_i\in \mathbb Z$. Suppose $\phi$ is not a positive and separable combination of the monotonicity axioms.
  The previous claim entails that the algorithm returns false. Consider again the graph $G$ and the sets $S_0,S_1$ after termination of the algorithm.
  Note that $S_1$ is now non-empty.
  Let $\vect{V}$ denote the set of variables appearing in $\phi$.
  Let $\calY$ be the (non-empty) collection of sets $\vect{Y}\subseteq \vect{V}$ such that for some  $j$, $(\vect{Y},j)$ belongs to $S^-$ and is connected to $S_1$. 
  Consider also its \emph{upper closure} $\calY^\uparrow \coloneqq \{\vect{Z} \subseteq \vect{V} \mid \exists \vect{Y} \in \calY: \vect{Y} \subseteq \vect{Z}\}$.  
  Define a mapping $h$ such that $h(\vect{Z})=1$ if $\vect{Z} \in \calY^\uparrow$, and otherwise $h(\vect{Z})=0$.
%
   Clearly, $h$ is a Boolean, monotone function. We  show
  that $h$ does not satisfy $\phi$. 
  
  Consider a pair $(\vect{X},i)\in S^+$
  such that $h(\vect{X})=1$. Then, $\vect{X}$ contains a set $\vect{Y}$ from $\cal{Y}$.
  Let $j$ be such that
   $(\vect{Y},j)$ belongs to $S^-$  and is connected to $S_1$.
  Since there is an edge from $(\vect{X},i)$ to $(\vect{Y},j)$, it follows that
  $(\vect{X},i)$ is connected to $S_1$.
  Now, if $(\vect{X},i)$ belonged to $S_0$, the algorithm could not have terminated yet.
  Hence $(\vect{X},i)$ must belong to $S^+ \setminus S_0$.  Recall that
  the backward edges form a bijection from $S^-\setminus S_1$ to $S^+\setminus S_0$.
  In particular, $(\vect{X},i)$ is the target of a unique backward edge with a source node $(\vect{Z},k)$.
  Since $(\vect{X},i)$ is connected to $S_1$, it follows that $(\vect{Z},k)$ is also connected to $S_1$. This entails that $h(\vect{Z})=1$. 
  In particular, this shows that any $(\vect{X},i)\in S^+$ such that $h(\vect{X})=1$ is paired 
  by a backward edge with a unique $(\vect{Z},k)\in S^-\setminus S_1$ such that $h(\vect{Z})=1$. In addition, because
  $S_1$ is non-empty, there exists an element $(\vect{U},l)\in S^-\cap S_1$  such that $h(\vect{U})=1$. In particular, $(\vect{U},l)$ is not the source node of any backward edge.
  These observations entail that $h$ does not satisfy $\phi$.
  This proves the contraposition of $(2)\Rightarrow (3)$ w.r.t. $c_i\in \mathbb Z$.
  
This concludes the proof of the direction $(2)\Rightarrow (3)$.

\end{proof}
The following example demonstrates that Alg. \ref{alg:mono} correctly returns false on the submodularity axiom, as this axiom is not a consequence of monotonicity and non-negativity.
\begin{example}
The submodularity axiom ${X}{Y} + {X}{Z} - {X} - {X}{Y}{Z} \geq 0$ is not valid over monotone functions. This can be also seen by referring to Alg. \ref{alg:mono}.  The set representation is
$(S^+, S^-)$ where 
$S^+ = \{({X}{Y},1), ({X}{Z},1)\} \text{ and }
S^-=\{({X},1), ({X}{Y}{Z},1)\}$. Suppose at the first step the algorithm introduces a backward edge from $({X},1)$ to  $({X}{Y},1)$; the only other option is the symmetric scenario where it introduces an edge from $({X},1)$ to $({X}{Z},1)$. After the first step we have $S_0 =\{({X}{Z},1)\}$ and $S_1=\{({X}{Y}{Z},1)\}$. Then, no path exists from $S_0$ to $S_1$, since no forward edge points to $({X}{Y}{Z},1)$. The algorithm therefore terminates  returning false. Accordingly, the function that maps $XYZ$ to $1$ and all other sets to $0$ is monotone, Boolean-valued, and does not satisfy the aforementioned submodularity axiom.  
\end{example}

\end{document}